\documentclass[aps,floats,showpacs,amstex,amssymb,tightenlines,notitlepage,hidelinks,floatfix]{revtex4-1}

\usepackage{mathtools}
\usepackage{amsmath}
\usepackage{amsthm}
\usepackage{amsfonts}
\usepackage{amscd}
\usepackage{epsfig}
\usepackage{amssymb}
\usepackage{tabularx}
\usepackage{longtable}
\usepackage{calligra}
\usepackage{enumerate}
\usepackage{hyperref}
\usepackage{mathrsfs}
\usepackage{calligra}
\def\j2{\mathbf{J}^2}
\def\a{\alpha}
\def\b{\beta}
\def\g{\gamma}

\def\e{\epsilon}
\def\ve{\varepsilon}
\def\f{\mathcal{F}}

\def\d{\delta}

\def\L{\Lambda}

\def\w{\omega}

\def\p{\partial}

\def\s{\mathcal{S}}
\def\T{\mathcal{T}}

\def\td{\widetilde{D}}
\def\hd{\widehat{D}}

\def\A{{\cal A}}

\def\fn{\left. \Phi_n \right|_t }

\def\rn{Reissner-Nordstr\"om }

\newtheorem{thm}{Theorem}

\begin{document}
\title{Black hole nonmodal linear stability: even perturbations of  Reissner-Nordstr\"om}

\author{Gustavo Dotti and Juli\'an M.  Fern\'andez T\'{\i}o }
\affiliation {Facultad de Matem\'atica, Astronom\'{\i}a y
F\'{\i}sica (FaMAF), Universidad Nacional de C\'ordoba and\\
Instituto de F\'{\i}sica Enrique Gaviola, CONICET.\\ Ciudad
Universitaria, (5000) C\'ordoba,\ Argentina}

\begin{abstract}
This paper is a companion of [Phys.\ Rev.\ D {\bf 95}, 124041 (2017)] in which, 
following a program on black hole nonmodal linear stability initiated in Phys.\ Rev.\ Lett.\  {\bf 112} (2014) 191101, 
odd perturbations of the Einstein-Maxwell equations around a Reissner-Nordstr\"om (A)dS 
black hole were analyzed. Here we complete the proof of the nonmodal linear stability of this spacetime 
by analyzing the even sector of the linear perturbations.  We show that all the gauge invariant information in the metric and 
Maxwell field 
even perturbations 
is encoded in 
 two    spacetime scalars: ${\mathcal S}$, which 
 is a gauge invariant combination of  $\delta (C_{\alpha \beta \gamma \epsilon}C^{\alpha \beta \gamma \epsilon})$ 
and  $\delta (C_{\alpha \beta \gamma \delta} F_{\alpha \beta} F^{\gamma \delta})$, 
and ${\mathcal T}$, a gauge invariant combination of
 $\delta ( \nabla _\mu F_{\alpha \beta} \nabla^\mu F^{\alpha \beta })$ and 
 $\delta ( \nabla_{\mu} C_{\alpha \beta \gamma \delta} \nabla^{\mu}C^{\alpha \beta \gamma \delta})$. Here 
 $C_{\alpha \beta \gamma \delta}$
 is the Weyl tensor, $F_{\alpha \beta}$ the Maxwell field  and $\delta$ means 
first order variation. 
 We prove that $\mathcal{S}$ and $\mathcal{T}$ are are in one-one correspondence with gauge classes of even linear 
 perturbations, 
 and that the linearized Einstein-Maxwell equations imply that these scalar fields are pointwise bounded  on the outer static region. 
 \end{abstract}

\pacs{04.50.+h,04.20.-q,04.70.-s, 04.30.-w}

\maketitle

\tableofcontents

\section{Introduction}

The Einstein-Maxwell field equations with cosmological constant $\L$ 
\begin{align} \label{efe}
&G_{\a \b} + \L g_{\a \b} = 8 \pi T_{\a \b}, \\  \label{Tab}
&T_{\a \b} = \tfrac{1}{4\pi} \left( F_{\a \g} F_{\b}{}^{\g} -\tfrac{1}{4} g_{\a \b} F_{\g \d}F^{\g\d} \right),\\
& \nabla_{[\a} F_{\b \gamma]} =0, \label{max1}\\  
&\nabla^{\b} F_{\a \b} = 0, \label{max2}
\end{align}
admit the solution 
\begin{align} \label{rn}
  ds^2  &= -f(r) dt^2 + \frac{dr^2}{f(r)} + r^2 (d \theta^2 + \sin ^2 \theta \; d \phi^2 ), \\
  \label{max0}
F &= E_0 \; dt \wedge dr, \;\; E_0=\frac{Q}{r^2},
  \end{align}
where the norm $f(r)$ 
   of the Killing vector $\p / \p t$ in (\ref{rn}) is 
\begin{equation} \label{f}
f(r) = 1 - \frac{2M}{r} + \frac{Q^2}{r^2} - \frac{\L}{3} r^2.
\end{equation}
Note that $r$ has geometrical meaning: is (the square root of one forth of) the 
areal radius of the spheres of symmetry under $SO(3)$. Note also that 
 $f = g^{\a \b} \nabla_\a r \nabla_\b r$. \\

We assume $\Lambda \geq 0$.  $M>0$ and $Q$ are constants of integration, they  
correspond to  mass and  charge respectively and we assume that 
 their values are such that (\ref{rn}) is a non extremal black hole, that is 
\begin{equation} \label{ff}
f = -\frac{\Lambda}{3r^2} (r-r_i)(r-r_h)(r-r_c)(r+r_i+r_h+r_c),
\end{equation}
where $0<r_i<r_h<r_c$ are the inner, event and cosmological horizons respectively.\\

We are interested in proving the {\em non-modal linear stability} of the outer 
static region $r_h<r<r_c$ of 
the solution (\ref{rn})-(\ref{max0}) of the field equations (\ref{efe})-(\ref{max2}). This 
concept of stability was defined in \cite{Dotti:2013uxa,Dotti:2016cqy} and implies proving that:

\begin{enumerate}[i)]
\item there are  gauge invariant (both in the Maxwell and infinitesimal diffeomorphism senses) 
scalar fields from the spacetime ${\cal M}$ 
into  $\mathbb{R}$ that contain the same information as  the gauge class 
$[(\mathcal{F}_{\a \b}, h_{\a \b})]$ of the perturbation $(\mathcal{F}_{\a \b}, h_{\a \b})$. 
Here  $\mathcal{F}_{\a \b}=\d F_{\a \b}$ is the first order perturbation of the electromagnetic 
  field and $h_{\a \b} = \d g_{\a \b}$  is the metric perturbation. These scalar fields then 
  measure the distortion of the geometry and the Maxwell field and the perturbation fields  
 $h_{\a \b}$  and   $\mathcal{F}_{\a \b}$   in a given gauge  
can be obtained by applying a linear functional on them. 
\item The gauge invariant curvature  fields are pointwise bounded on the outer static region by constants that depends on 
the initial data of the perturbation on a Cauchy surface for that region.
\end{enumerate}
 For odd perturbations i) and ii) were proved  in 
the companion  paper \cite{julian}. In 
this paper we complete the proof of nonmodal linear stability of the charged black hole by considering 
 the even sector of the linear perturbation fields. \\

A discussion of the relevance of the nonmodal linear stability concept above, which was introduced in \cite{Dotti:2013uxa},  
can be found in Section I of \cite{Dotti:2016cqy}. For the Schwarzschild spacetime, the strategy behind the proof of 
nonmodal linear stability 
in \cite{Dotti:2013uxa} was using the  supersymmetric  even/odd duality to show that 
 {\em both odd  and even} linear gravity perturbation equations 
are equivalent to  (independent)
 four dimensional Regge-Wheeler equations. 
This also holds for Schwarzschild de Sitter, a detailed proof covering the $\L \geq 0$ 
 cases 
is given in   Lemma 
7 in \cite{Dotti:2016cqy}.) Once the linearized gravity problem is reduced to uncoupled  four dimensional 
scalar wave equations with a time 
independent 
potential, it is possible to place pointwise bounds on the geometric scalar fields mentioned above, and to  analyze their decay 
along 
future causal directions. This duality is of no use in the charged black hole case because the 
odd sector equations have the same level of complexity of those of the even sector and, and contrary to what happens 
in the $Q=0$ case, the set of odd mode equations is {\em not} equivalent to a four dimensional scalar field equation with 
a time independent potential. \\

Our emphasis in this series of papers is on finding the appropriate set of gauge invariant, curvature related scalar fields 
 encoding 
 the information of the gauge class 
of the perturbation; we do  not analyze their decay. \\

 We leave aside the asymptotically AdS $\L<0$ case. We do so because the dynamics of perturbations is non-unique 
in this case ---in particular, the notion of stability is ambiguous--- due to the conformal timelike boundary. In this case also,
 a choice of 
boundary conditions  at the conformal boundary generically breaks the even/odd duality, so that  the even sector perturbation 
equations  are {\em not} equivalent in (even in the uncharged case) to 
a four dimensional Regge-Wheeler equation, as happens for $\L \geq 0$
(for further details see Section IV in \cite{bernardo}.)\\

As in \cite{Kodama:2003kk}, the warped structure of the spacetime (\ref{rn})  
  ${\cal M} = {\cal N} \times_{r^2} \sigma$ 
  \begin{equation} \label{warped} 
g_{\a \b} dz^{\a} dz^{\b} = \tilde g_{ab}(y) \, dy^a dy^b + r^2(y) \hat g_{AB}(x)\,
  dx^A dx^B.
\end{equation}
is used    to simplify the linearized Einstein Maxwell equations (LEME. We also use the acronyms 
LEE for linearized Einstein equations and LME for linearized Maxwell equations). The  
``orbit manifold" ${\cal N}$ is  two dimensional and Lorentzian, with  line element 
$\tilde g_{ab}(y) dy^a dy^b$ ($= -f dt^2 + \frac{dr^2}{f}$ in Schwarzschild coordinates);  
the ``horizon manifold''  $\sigma$ with metric  $\hat g_{AB}(x)\,  dx^A dx^B$ is the unit two sphere (for 
a treatment of linearization around warped metrics in arbitrary dimensions and with constant curvature 
horizon manifolds  see \cite{Kodama:2003kk} 
and references therein.)  
In (\ref{rn}), $(t,r)$ coordinates are used for ${\cal N}$ and the standard angular coordinates 
$ \hat g_{AB}(x)\,  dx^A dx^B = d \theta^2 + \sin ^2 \theta \; d \phi^2$ are used for the unit sphere. In what follows 
our treatment is ``2D-covariant'', that is, it allows independent coordinate changes in ${\cal N} $ and the unit sphere. \\

Equation (\ref{warped}) illustrates our notation,  which we adopted from \cite{Chaverra:2012bh}: we use 
  lower case indexes $a, b, c, d, e$ for tensors on the orbit manifold ${\cal N}$, 
upper case indexes   $A, B, C, D,...$ 
for tensors on $S^2$, and Greek indexes  for space-time tensors. We  follow the additional  convention 
 in \cite{Dotti:2016cqy} that 
\begin{equation} \label{index}
\a=(a,A), \b=(b,B), \g=(c,C), \d=(d,D).
\end{equation}
Tensor fields {\em   introduced}
 with a lower $S^2$ index (say $Z_A$) and {\em then shown with an upper $S^2$ index} are assumed 
to have been acted  upon {\em with the unit $S^2$ metric inverse $\hat g^{AB}$},
 (i.e., in our example, $Z^A \equiv \hat g^{AB} Z_B$), 
and similarly with upper $S^2$ indexes moving down.  This has to be kept in mind to avoid wrong $r^{\pm2}$ factors 
in the equations. 
$\td_a, \tilde \e_{ab}$ and $\tilde g^{a b}$ are the covariant derivative, volume form (any chosen orientation) and metric inverse 
for the $\mathcal{N}$ orbit space; $\hd_A$ and  $\hat \e_{AB}$ are the covariant derivative and volume form $\sin (\theta) d\theta \wedge d\phi$ 
on the unit sphere. \\

The metric and Maxwell field perturbations $h_{\a \b}$ and   $\mathcal{F}_{\a \b}$ admit a series expansion  in   
 rank 0,1, and 2 eigentensor fields of the horizon manifold Laplace-Beltrami  (LB) operator,  with ``coefficients'' that are  
tensor fields on the orbit space $\mathcal{N}$
 \cite{Kodama:2003kk}. Individual terms of this series are called ``modes'', they are not mixed by 
 the LEME. In the standard modal approach a master scalar field  
$\mathcal{N} \to \mathbb{R}$ is extracted for each mode and the LEME is  
reduced to an infinite set of {\em scalar wave equations on} $\mathcal{N}$ (that is, 1+1 wave equations), one for each master 
mode. Modal stability consists in proving the boundedness/decay of these master fields. 
This was proved in four dimensional General Relativity 
in the seminal black hole stability papers \cite{Regge:1957td} \cite{Zerilli:1970se} \cite{Zerilli:1974ai}  and in higher 
dimensions more recently by Kodama and 
Ishibashi  (see, eg.g. \cite{Kodama:2003kk} and \cite{Ishibashi:2011ws}). 
All notions of linear stability prior to \cite{Dotti:2013uxa} were {\em modal}, that is, restricted  to 
  the boundedness of the 1+1 master fields.   For four dimensional charged black holes 
  the modal linear stability in the case  $\L=0$ 
  was  proved by Zerilli and  Moncrief  in the series 
of articles \cite{Zerilli:1974ai, Moncrief:1974gw, Moncrief:1974ng, Moncrief:1975sb} (see also \cite{wald}) \\

The limitations of the modal linear stability  are explained   in \cite{Dotti:2013uxa} and 
\cite{Dotti:2016cqy} (see the Introduction of \cite{Dotti:2016cqy} for a detailed explanation). 
These two papers are devoted to 
 the nonmodal linear stability of the Schwarzschild and Schwarzschild de Sitter black hole. 
The nonmodal linear stability of the \rn black hole with $\L \geq 0$, under odd perturbations,  
was established in \cite{julian}. 
In the following Sections we complete the proof of nonmodal linear stability of this black hole by proving its  stability 
under even perturbations. 
We do so by showing  that there are two fields made out of gauge invariant first order perturbations of curvature scalars 
(for details refer to Section \ref{invs}). These fields 
encode all the gauge invariant information of arbitrary even perturbations, allow to reconstruct 
the metric and Maxwell field perturbations in a given gauge, and are pointwise bounded.

\section{Linearized Einstein-Maxwell equations} \label{rwz}

The LEME are obtained by linearizing equations  (\ref{efe})-(\ref{max2}), that is, we assume
  that there is a smooth   one-parameter set 
of solutions $(g({\ve})_{\a \b}, F(\ve)_{\a \b})$ 
of the Einstein-Maxwell equations (\ref{efe})-(\ref{max2}) such that  $ (g({\ve=0})_{\a \b}, F(\ve=0)_{\a \b})$ are the 
 Reissner-Nordstr\"om  fields  (\ref{rn})-(\ref{max0}), take the derivative with respect to $\ve$ and evaluate it 
 at $\ve=0$. The resulting  equations are linear in  the {\em perturbation fields} $h_{\a \b}:= d g_{\a \b} / d\ve |_0$ 
 and  ${\mathcal{F}}_{\a \b}:=d F_{\a \b} / d\ve |_0$.  \\

The linearization of equation (\ref{max1}) gives $d \mathcal{F}=0$. Since the region we are interested in 
(see Theorem \ref{2} for details) is homeomorphic to $\mathbf{R}^2 \times S^2$, then of the same homotopy type of $S^2$, 
$d \mathcal{F}=0$  implies that there exists $A_{\a}$ such that 
\begin{equation} \label{dF}
\mathcal{F}_{\a \b} = \p_{\a} A_{\b} - \p_{\b} A_{\a} + p \; \hat \e _{\a \b},
\end{equation}
where $p$ is a constant and $\hat \e _{\a \b}$ is the pullback to $\mathcal{M}$ of the $S^2$ volume form $\hat \e _{A B}$. 
 Under the index 
convention (\ref{index}) the covector field $A_{\a}$ is written as 
\begin{equation} \label{vd}
A_{\a}=(A_a, A_A) 
\end{equation}
and, as explained in \cite{Dotti:2016cqy,Ishibashi:2004wx},  admits a decomposition 
in a set of even (+) and  odd (-) fields: 
\begin{equation} \label{Adecomp}
A_{\a}=(A^+_a, \hd_A A^+ + \hat \e_{A}{}^C \hd_C A^-), 
\end{equation}
 Even and odd fields  are characterized by the way they transform 
   when pull backed by the 
antipodal map $P$ on $S^2$ \cite{Dotti:2013uxa}. 
Note that $\hat \e _{\a \b}$ is an odd field, and that  equations (\ref{dF})-(\ref{Adecomp}) imply that we can replace 
\begin{equation} \label{adecomp}
\mathcal{F}_{\a \b} \text{ with }\{ A_{a}^+, \; A^+ \} \cup \{ A^- , p \},
\end{equation}
The constant $p$ associated with the  odd Maxwell field perturbation $ p \; \hat \e _{\a \b}$ corresponds 
to turning on a magnetic charge \footnote{This odd perturbation was disregarded in the companion paper 
\cite{julian}. It produces no backreaction {\em at first order}: $Q^2$ in (\ref{f}) has to be replaced with 
$Q^2+p^2$.}.
The scalar fields $A^{\pm}$ are unique if they are required to belong to $L^2(S^2)_{>0}$ \cite{Dotti:2016cqy}. 
Here $L^2(S^2)_{>\ell_o}$
is the  space 
of square integrable functions on $S^2$ orthogonal to the $\ell=0, 1,...,\ell_o$ eigenspaces of the Laplace-Beltrami (LB) 
 operator, and 
 $\ell$ labels the LB scalar field eigenvalue $-\ell (\ell+1)$. \\

Similarly, a symmetric tensor field
 $S_{\a \b} = S_{(\a \b)}$, such as $h_{\a \b}$,  $\mathcal{G}_{\a \b} := d G_{\a \b} /d\ve |_0$  
and $ \mathcal{T}_{\a \b} :=d T_{\a \b}/ d\ve |_0$,  decomposes as \cite{Ishibashi:2004wx, Dotti:2016cqy,julian}
\begin{equation} \label{std1}
S_{\a \b} = \left( \begin{array}{cc} S_{ab} & S_{a B} \\ S_{A b} & S_{AB} \end{array} \right),
\end{equation}   
with
\begin{equation} \label{std2}
S_{a B} = \hd_B S_a^+ + \hat \e_{B}{}^C \hd_C S_a^-.
\end{equation}
Assuming that  $S_a^{\pm} \in L^2(S^2)_{>0}$, they 
are unique \cite{Dotti:2016cqy,Ishibashi:2004wx}. $S_{AB}=S_{(AB)}$ 
further decomposes as 
\begin{equation} \label{std3}
S_{AB} = \hd_{(A} ( \e_{B) C} \hd^C S^-) + \left( \hd_A \hd_B - \frac{1}{2} \hat g_{AB} \hd^C \hd_C \right) S^+ + \tfrac{1}{2} \; S_T^+ \; \hat g_{AB} , 
\end{equation}
where $S_T^+ = S_C{}^C$
and the fields   $S^{\pm} \in L^2(S^2)_{>1}$ are unique.  \\

 In this way, as happens for covector fields (equation (\ref{adecomp})), the symmetric tensor field $S_{\a \b}$ 
is replaced by a set of even  and  odd fields
\begin{equation} \label{std4}
\{ S_{ab}^+=S_{ab}, \; S_a^+, \; S^+, \; S_T^+ \} \cup \{ S_a^-, \; S^- \}.
\end{equation}
In particular, the perturbed metric, Einstein tensor and energy momentum tensors contain the  fields 
\begin{align} \label{hdecomp}
{h}_{\a \b} &\sim \{  h_{ab}^+, \; h_a^+, \; h^+, \; {h}_T^+ \} \cup \{ h_a^-, \; h^- \}\\
\label{Gdecomp}
\mathcal{G}_{\a \b} &\sim \{  G_{ab}^+, \; G_a^+, \; G^+, \; G_T^+ \} \cup \{ G_a^-, \; G^- \}\\
\mathcal{T}_{\a \b} &\sim \{  T_{ab}^+, \; T_a^+, \; T^+, \; T_T^+ \} \cup \{ T_a^-, \; T^- \}. \label{Tdecomp}
\end{align}
Even and odd fields  are not mixed by the LEME.   The restriction of the LEME to the odd sector was 
the subject of \cite{julian}, even perturbations are studied in the following sections. \\

Let $J_{(1)}, J_{(2)}$ and $J_{(3)}$ be $S^2$ (and therefore
 spacetime) Killing vector fields corresponding to rotations 
around orthogonal axis in $\mathbb{R}^3 \supset S^2$, normalized such that the length of their 
closed orbits in the unit sphere  is   $2 \pi$ 
(e.g., $J_{(3)} = \p / \p_{\phi}$).
The square angular momentum operator
\begin{equation} \label{j2}
\j2 \equiv (\pounds_{J_{(1)}})^2 + (\pounds_{J_{(2)}})^2 + (\pounds_{J_{(3)}})^2, 
\end{equation}
 is defined both in $S^2$ and the spacetime.  
This operator 
commutes with the LEME and preserves parity. It thus allows a
 further decomposition of even and odd fields  into modes (eigenfields of $\j2$). 
 On $S^2$ scalars the operator $\mathbf{J}^2$ agrees with the  LB 
 operator of $S^2$, 
$\hd^A \hd_A$; however, on higher rank tensors  these two operators 
act differently. Since  $[\nabla_a, \pounds_{J_k}] = 0 = [\hd_A, \pounds_{J_k}] = [\td_a, \pounds_{J_k}]$, it 
follows that $\mathbf{J}^2$ 
   commutes 
with $\nabla_{\a}$, $\td_a$ and $\hd_A$. 
In a  modal decomposition approach 
 the tensor fields on the right sides of (\ref{hdecomp})-(\ref{Tdecomp})  
into eigenfields of $\j2$. \\

In the following sections we restrict ourselves to even perturbations and assume the restrictions above: 
$A^{+}, S_a^{+} \in L^2(S^2)_{>0}$, $S^{+} \in L^2(S^2)_{>1}$. These conditions   guarantee 
that  the  linear operators $(A_a^+,A^+) \to A_{\a}$ in (\ref{Adecomp}), 
and $\{ S_{ab}^+, \; S_a^+, \; S^+, S_T^+ \}  \to S_{\a \b}$ 
in (\ref{std1})-(\ref{std3}) are injective (Lemma 2 in \cite{Dotti:2016cqy}). 
Since we restrict to even perturbations, there is no risk of confusion and  + superscripts will be suppressed 
from now on.

\subsection{Even sector perturbations} \label{osp}

Even  perturbations are those for which the minus fields in (\ref{Adecomp})  and (\ref{hdecomp}) are  zero. 
Re-scaling and dropping the + superscripts, $h_T^+ =: r^2 h_T$, $h^+ =:2r^2 h$,  gives 
\begin{equation} \label{pert1} \renewcommand*{\arraystretch}{1.3}
h_{\a \b} = \left( \begin{array}{cc} h_{ab} & \;\hd_B h_a \\ \hd_A h_b & 
\;r^2\left[\left( 2\hd_A \hd_B - \hat g_{AB} \hd^C \hd_C \right) h + \tfrac{1}{2} \; h_T\; \hat g_{AB} \right]\end{array} \right) 
\end{equation}
with the restrictions $h_a \in L^2(S^2)_{>0}$, $h \in  L^2(S^2)_{>1}$. \\
Similarly, equations (\ref{dF}) and the even piece of (\ref{Adecomp})  give 
(dropping  superscripts) 
\begin{equation} 
\f_{\a \b} = \left( \begin{array}{cc} \td_aA_b - \td_bA_a &\;  \td_a \hd_B A - \hd_B A_a \\
  \hd_A A_b -\td_b \hd_A A &\;  0
 \end{array} \right) \label{pert2}
\end{equation}  
with $A \in L^2(S^2)_{>0}$. \\

$U(1)$ gauge transformations of the Maxwell field leave $\f_{\a \b}$ invariant while changing the potential as 
$A_{\a} \to A_{\a } + \p_{\a} B$.  
The even piece of the vector potential (\ref{Adecomp}) then changes as $A_a \to A_a + \p_a B$ and  
$A \to A + B_{(>0)}$, where $B_{>0}$ is the projection of $B$ onto  $L^2(S^2)_{>0}$. 
\\

Under a {\em coordinate} gauge transformation (infinitesimal diffeomorphism)  along the even  vector 
field defined by 
\begin{equation} \label{gtcv}
X_{\a}=(X_a, r^2 \hd_A X), \;\;\;  
X \in L^2(S^2)_{>0},
\end{equation}
 $h_{\a \b}$ and $\f_{\a \b}$ transform into the  physically equivalent  fields:
\begin{equation} \label{prime}
h'_{\a \b} = h_{\a \b} + \pounds_{X} g_{\a \b}, \;\; \f'_{\a \b}=\f_{\a \b} + \pounds_{X} F_{\a \b}.
\end{equation}
From (\ref{rn}), (\ref{max0}),  (\ref{pert1}), (\ref{pert2}), (\ref{gtcv})  and (\ref{prime}) we find that   (\ref{prime}) 
 is equivalent  to 
\begin{equation} \label{gauge}
\begin{rcases}
&h_{ab}   \to h'_{ab}=h_{ab} + \td_a X_b +\td_b X_a, \\ 
& h_T   \to h_T' = h_T + \frac{4}{r}\td^a r\; X_a + 2\hd^C \hd_CX \;\;\; \\
&A_b \to A_b' =  A_b -  \widetilde \e_{bc} X^c E_0\\
\end{rcases} \text{ all } \ell
\end{equation}
where the legend ``all $\ell$" reminds us that these fields have projections on all the $\ell$ subspaces, whereas 
\begin{equation}\label{gauge>0}
\begin{split}
& A \to A'=A, \;\; (\ell>0 \text{ only}),\\
  & h_a  \to h_a'= h_a + X^{>0}_a + r^2\td_aX, \;\; (\ell >0 \text{ only}),\\ 
&h  \to h'= h + X_{>1}, \;\; (\ell >1 \text{ only}).
 \end{split}
\end{equation}

\subsubsection{$\ell=0$: solution of the LEME and linearized Birkoff theorem}

Given that $\ell=0$ correspond to the spherically symmetric part of the perturbation,  
$\ell=0$ perturbations to the spherically  symmetric \rn background that solve the LEME should amount, 
in view of Birkhoff's theorem,  to a  modification of the parameters $Q$ and $M$ 
in (\ref{rn})-(\ref{f}). In this Section we prove that this is the case.\\

On $\ell=0$, the fields $h_a$,  $h$, $A$ and $X$  have trivial projections, and we can use (\ref{gauge})  
as in \cite{Dotti:2016cqy}, choosing $\td^arX^{(\ell=0)}_a=-\frac{r}{4}h_T^{(\ell=0)}$ 
to set $h_T'=0$ and then 
$2\td^a X^{(\ell=0)}_a = -g^{ab}  h_{ab}^{(\ell=0)}$ to get a traceless $h'_{ab}$. 
Dropping primes, the resulting metric perturbation is of the form 
\begin{equation} \label{rwhL=0}
h_{\a \b}^{(\ell=0)} = \left( \begin{array}{cc} h_{ab}^{T,(\ell=0)} & \;\;0 \\ 0 & 
\;\;0 \end{array} \right), \;\;\; \;\;\; \tilde{g}^{ab} h_{ab}^{T,(\ell=0)}=0.
\end{equation} 
This gauge choice admits  a  
 residual freedom $X_\a=(X_a,0)$ preserving the conditions   (\ref{rwhL=0}), for which $X_a$ must satisfy
\begin{equation}\label{regauge}
 (\td^ar)X_a=0,\;\;\; \;\td^a X_a=0,
\end{equation}
whose solution is 
\begin{equation} \label{X(r)}
 X_a=\widetilde{\e}_{ab}\td^b X(r).
\end{equation}
Since the $\ell=0$ piece of $A$ is trivial, the  $\ell=0$ Maxwell field is 
\begin{equation}  \label{l=0F}
\f_{\a \b}^{(\ell=0)} = \left( \begin{array}{cc} \td_aA_b^{(\ell=0)} - \td_bA_a^{(\ell=0)} &\;  0 \\
  0 &\;  0
 \end{array} \right) 
\end{equation}  
and the linearization of (\ref{max2}) reduces to
\begin{equation}
 \td^b\left(r^2\left(\td_aA_b^{(\ell=0)}-\td_bA_a^{(\ell=0)}\right)\right)=0.
\end{equation}
Defining $\td_aA_b^{(\ell=0)}-\td_bA_a^{(\ell=0)}=:\widetilde{\e}_{ab}\mathcal{E}^{(\ell=0)}$ the above equation reads 
\begin{equation}
\widetilde{\e}_{ab}\td^b\left(r^2\mathcal{E}^{(\ell=0)}\right)=0.
\end{equation}
Its  solution,
\begin{equation}
 \mathcal{E}^{(\ell=0)}=\frac{q}{r^2}, \label{Maxper}
\end{equation}
corresponds to a change   in  charge $Q \to Q + \epsilon q$, as anticipated.  \\

To complete our proof of the ``linearized Birkoff theorem'' we choose coordinates $(t,r)$ in orbit 
space,  work in the transverse gauge (\ref{rwhL=0}) and  use the residual gauge freedom (\ref{X(r)}) 
to set $h_{tr}=0$ (this fixes $X(r)$ in (\ref{X(r)})  up to a linear function of $r$). 
Using this additional 
condition together with the trace-free condition $h_{tt}= f^2 h_{rr}$ and $A_t=q/r$, $A_r=0$,  the $t-r$ 
 component of  the LEE 
\begin{equation} \label{LEME1}
 {\mathcal{G}_{\a\b}}^{(\ell=0)}+\Lambda {h_{\a\b}}^{(\ell=0)}=8\pi\mathcal{{T}}^{(\ell=0)}_{\a\b}\\
\end{equation}
gives $\p_t h_{rr}=0$, so that  $h_{rr}$ and $h_{tt}= f^2 h_{rr}$ depend only on $r$. 
Inserting this condition in the $r-r$ LEE (\ref{LEME1})  gives
\begin{equation} \label{hrr0}
h_{rr}= - \frac{2r^2(qQ-mr)}{r^4 f^2}, 
\end{equation}
where $m$ is a constant of integration. We conclude that 
\begin{equation} \label{htt0}
h_{tt} = f^2 h_{rr} = -\frac{2qQ}{r^2} + \frac{2m}{r}. 
\end{equation}
Note that (\ref{hrr0}) and (\ref{htt0}) 
 correspond precisely to, respectively, $(m \p/\p_M + q \p/\p_Q) f^{-1}$ and  $(m \p/\p_M + q \p/\p_Q) (-f)$, so we 
 recognize that $m$ and $q$  correspond respectively to   
 first order variations  $\d M$ and $\d Q$ of the mass and charge in the background \rn metric. 

\subsubsection{$\ell=1$ modes: gauge choice} 

Using the gauge freedom (\ref{gauge}) we can put the metric perturbation in Regge-Wheeler (RW) form:
\begin{equation}  \renewcommand*{\arraystretch}{1.3}
{}^{RW}h_{\a \b}^{(\ell=1)} = \left( \begin{array}{cc} h_{ab}^{(\ell =1)} &0 \\ 0 & 
\;\tfrac{r^2}{2} \; \hat g_{AB} \; h_T^{(\ell = 1)} \end{array} \right) \label{rw=1}
\end{equation} 
Contrary to what happens for $\ell>1$, for $\ell=1$ there is no unique RW gauge: once the metric is put in  RW form 
(\ref{rw=1}),  
we can gauge transform it into a {\em different} RW gauge using a gauge vector  
of the form $X_\a=(X_a,\hd_AX)$ 
with 
\begin{equation}\label{rgfL=1a}
X^{(\ell=1)}_c=-r^2\td_cX^{(\ell=1)}.
\end{equation}
We will use this  gauge freedom to further set 
\begin{equation} \label{thl1}
\tilde h^{(\ell=1)} := \tilde g^{ab} h_{ab}^{(\ell=1)}=0.
\end{equation}
We will assume the RW traceless gauge conditions (\ref{rw=1}) and (\ref{thl1}) when solving the LEME. 
Note that this does not exhaust the gauge transformations   (\ref{rgfL=1a}): a residual gauge freedom 
keeping these conditions is one for which the gauge vector 
satisfies  (\ref{rgfL=1a}) together with 
\begin{equation} \label{rgfL=1b}
 \td^c(r^2 \td_c X^{(\ell=1)})=0.
\end{equation}

\subsubsection{$\ell \geq 2$ modes: gauge choice and gauge invariants} \label{slgeq2}

For $\ell \geq 2$ the field
\begin{equation}
p_a=h^{(\geq 2)}_a-r^2\td_a h
\end{equation}
transforms as $p_a \to p_a' = p_a + X_a^{(\geq 2)}$. This 
allows to construct the following ($\ell \geq 2$) gauge invariant fields 
(we use (\ref{pert2})-(\ref{gauge>0})): 
\begin{equation}
\begin{rcases}\label{ginvs}
H_{ab} := h^{{(\geq 2)}}_{ab} - \td_a p_b - \td_b p_a \\
H_T := h^{(\geq 2)}_T - \frac{4}{r}p_a\td^ar-2\hd^C \hd_Ch^{(\geq2)} \\
\mathcal{E}\;\widetilde{\e}_{ab}:= \f_{ab}^{\geq 2} - \widetilde{\e}_{ab}\td_c \left(E_0p^{c}\right)  \\
r \widetilde{\e}_{ab}\;\hd_B\mathcal{E}^b := \f_{aB}^{\geq 2}-\widetilde{\e}_{ab}E_0\td_Bp^{b}
\end{rcases} \ell \geq 2 \text{ gauge invariant fields}
\end{equation}
The RW gauge is defined by the condition 
$p_a=0$. It is unique, since any nontrivial  gauge transformation 
(\ref{gauge})-(\ref{gauge>0})  
requires $X \neq 0$ to keep $h_a=0$, and this spoils the condition  $h=0$. 
\begin{equation}  \renewcommand*{\arraystretch}{1.3} 
{}^{RW}h_{\a \b}^{\geq 2} = \left( \begin{array}{cc} H_{ab} & 0 \\ 0 & \frac{r^2}{2}\hat g_{AB} H_T  
\end{array} \right). \label{pertrw1}
\end{equation} 
Note that this is formally identical to (\ref{rw=1}). 

\subsubsection{Recasting the linearized $\ell \geq 2$ equations}

In what follows we will decompose $S^2$ and orbit space symmetric 2-tensors into their traceless a 
pure trace pieces as
\begin{align}\label{deco}
S_{ab} &= S_{ab}^T + \tfrac{1}{2} g_{ab} \tilde S, \;\;\;\tilde S := S_{ab} g^{ab}\\
S_{AB} &= S_{AB}^T + \tfrac{1}{2} \hat g_{AB} \widehat S, \;\;\; \widehat S :=S_{AB} \hat g^{AB}.
\end{align}
We will assume the linearized $\ell \geq 2$ 
Maxwell  field is given by (\ref{pert2}) and that the linearized $\ell \geq 2$ metric 
is in RW form  (\ref{pertrw1}). \\

Consider first the linearized Maxwell equations. 
Equations  (\ref{pert2}) and (\ref{pertrw1}) imply  that the $\beta=B$ components 
of the linearization of the Maxwell equation $\nabla^{\a}F_{\a \b}$ are equivalent to the condition 
\begin{equation} \label{c1}
\hd_B(\td^d A_d -\td^d \td_d A)=0,
\end{equation}
 which can be written as 
\begin{equation}
\widetilde{\e}^{ab}\td_a \left(r\mathcal{E}_{b}\right)=0, \;\; \; \mathcal{E}^b := r^{-1} \; \widetilde{\e}^{bc}(\td_c A-A_c).  
\label{mx2}
\end{equation}
This implies that
\begin{equation} \label{calA}
\mathcal{E}_b=-\frac{1}{r}\td_b \mathcal{A},
\end{equation}
for some scalar $\mathcal{A}$,  
and simplifies (\ref{pert2}) to
\begin{equation} \label{fab2}
\f_{\a \b} = \left( \begin{array}{cc}- \widetilde{\e}_{ab}\td_c\td^c\mathcal{A} & -\widetilde{\e}_{a}{}^c \td_c \hd_B
\mathcal{A} \\
\widetilde{\e}_{b}{}^c \td_c\hd_A \mathcal{A} & 0  \end{array} \right)
\end{equation} 
The 
  $\beta=b$ components of the linearization of  $\nabla^{\a}F_{\a \b}$ 
  then gives $E_0 \widetilde \e_a{}^b \p_b z=0$, 
  where 
  $ z:= -\tfrac{r^2}{Q}\td_a\td^a \A-\tfrac{1}{Q}\hd_A\hd^A \A-\frac{1}{2}\left(H_{ab}g^{ab}-h_T\right)$. 
  This gives 
  $z=z(\theta,\phi)$. However, in view of the $U(1)$ gauge freedom freedom $\A \to \A' = \A+ p(\theta,\phi)$ implicit in the definition 
  (\ref{calA}) of $\A$, and given that $z$ has no $\ell=0$ component, we can choose $p$ such that $\hd^B \hd_B p =Q z$, then for $\A'$ 
  we find $z'=0$ and (dropping the prime on $\A$)
\begin{equation}
\td_a\td^a \A+\frac{1}{r^2}\hd_A\hd^A \A=\frac{Q}{2r^2}\left(h_T - \widetilde{H}\right). \label{max}
\end{equation}
where $\widetilde{H}$ denotes the trace part of $H_{ab}$ according to (\ref{deco}). \\

From now on  we  switch from  $H_{ab}^T$ to the one form $C_a = H_{ab}^T \td^b r$, which contains  the same 
information, in view of the equality
\begin{equation}
H_{ab}^T = \frac{1}{f} \left( \td_a r \; C_b + C_a \; \td_b r - g_{ab} \; \td^d r \; C_d \right).
\end{equation}

Having solved the LME we proceed with the  LEE. The traceless  $S^2$ piece   
\begin{equation}
\mathcal{G}_{AB}^{T}+\Lambda h_{AB}^{T}=8\pi\mathcal{{T}}^{T}_{AB} \label{LEME11}
\end{equation}
 gives 
 \begin{equation} \label{th}
 \tilde H =0.
\end{equation}
The off-diagonal piece
\begin{equation}
\mathcal{G}_{Ab}+ \Lambda h_{Ab} = 8\pi\mathcal{{T}}_{Ab},  \label{LEME13}
\end{equation}
combined with the condition (\ref{th}) (and $h_{Ab}=0$), 
gives 
\begin{equation} \label{qb}
q_b:= \td_a  C^a \; \td_b r+ \widetilde{\e}^{ec} \td_e C_ c \; \widetilde{\e}_b{}^a \td_a r -\frac{f}{2}\td_bh_T-4fE_0\td_b \mathcal{A}=0.
\end{equation}
Contracting (\ref{qb})  with $\widetilde{\e}^{bd} \td_d r$ gives
\begin{equation}
\widetilde{\e}^{bd} \left[ \td_b C_d + \tfrac{1}{2} \td_b h_T \; \td_d r + 4E_0 \; \td_b \mathcal{A} \; \td_d r \right] =0.
\end{equation}
This allows to introduce the field $\xi$,  defined by
\begin{equation} \label{xi}
\td_d \xi = Z_d:= C_d - \tfrac{1}{2}  r \; \td_d h_T + 4  E_0 \;  \mathcal{A} \; \td_d r.
\end{equation}
Contracting (\ref{qb})  with $\td^b r$ and using the above equation then gives
\begin{equation} 
\td^a \td_a \xi +\frac{r}{2}\td^a\td_ah_T-8E_0(\td^ar)\td_a \mathcal{A}-4E_0 \mathcal{A} 
\td^a\td_ar+\frac{8E_0}{r}(\td^ar)(\td_ar) \mathcal{A}=0. \label{E21}
\end{equation}
Using equations (\ref{max}), (\ref{th}), (\ref{qb}) and (\ref{E21}) in the LEE
\begin{equation} \label{LEME12}
\widetilde{ \mathcal{G}}+\Lambda \widetilde{H}=8\pi \widetilde {\mathcal{T}} 
\end{equation}
gives  
\begin{equation}
\frac{2}{r} \td^a \td_a \xi - 8 E_0 \left( \frac{(\td^a \td_a r) \mathcal{A}}{r} + \frac{\hd^A \hd_A \mathcal{A}}{r^2} \right) 
+ \frac{4}{r^2} \td^a \xi \td_a r  = \left[ \left(\frac{\hd^A \hd_A +2}{r^2}\right) - 4 E_0^2 \right] h_T.
\end{equation}
Note that the operator on the right side above is invertible, this proves that all components of $\f_{\a \b}$ and $h_{\a \b}$ 
can be written  in terms of $\xi$ and $\mathcal{A}$. If we do so and use the remaining LEE, we arrive, after some work,  
to  
 the following system  of partial differential equations 
for $\xi$ and $\A$:
\begin{multline} \label{3.24}
 \left(\hd_A\hd^A+2f-r\td_a\td^ar\right) \left[\td_a\td^a\xi+\frac{2}{r}\td_a\xi \td^a r 
- \left(\hd_A\hd^A+r\td^a\td_ar\right)\left(\frac{4Q\mathcal{A}}{r^3}\right) \right]
\\+\left(\hd_A\hd^A+2-\frac{4Q^2}{r^2}\right)\left(\frac{1}{r^2}\hd_A\hd^A\xi-\frac{2}{r}\td_a\xi \td^a r +\frac{8fQ}{r^3}
\mathcal{A}\right) =0
\end{multline}

and

\begin{equation}\label{3.23}
\left(\hd_A\hd^A+2f-r\td_a\td^ar\right)\left(\td_a\td^a\A+\frac{1}{r^2}\hd_A\hd^A\mathcal{A}
\right) + \frac{8fQ^2}{r^4} \mathcal{A}
+\frac{Q}{r}\left(\frac{1}{r^2}\hd_A\hd^A\xi-\frac{2}{r}\td_a\xi \td^a r\right)=0 
\end{equation}
Note that, since they are derived from gauge invariant fields, $\A$ and $\xi$ above are gauge invariant.

\subsubsection{Solution of the $\ell=1$ LEME}

Equation (\ref{rw=1}) is formally identical to (\ref{pertrw1}),  the difference being that 
the latter is given in terms of gauge invariant fields. Thus, the steps (\ref{c1}) to (\ref{max}) from the previous 
Section hold for $\ell=1$ with the replacements $H_{ab}\to h_{ab}^{(\ell=1)}$, etc. 
Now, in view of equation (\ref{std3}), equation (\ref{LEME11}) is void for $\ell=1$. However, the trace free condition 
(\ref{th}) to where this equation leads  corresponds to the  traceless {\em gauge choice } 
(\ref{thl1})  for $\ell=1$. This implies that the reasoning
 following (\ref{th}) can also be taken without change for $\ell=1$. As a result, we obtain the system 
(\ref{3.24})-(\ref{3.23}) with $\A \to \A^{(\ell=1)}$, $\xi \to \xi^{(\ell=1)}$ (defined in a way analogous to 
 (\ref{xi})), and $\hd^A \hd_A 
\to -2$. \\

A conceptual difference between the $\ell=1$ and $\ell >1$ cases is that the fields
$\A$ and $\xi$, being defined from the gauge invariant fields, are themselves gauge invariant, whereas  $\A^{(\ell=1)}$ and 
$\xi^{(\ell=1)}$ are {\em not}. Tracing 
back the gauge transformations of the fields involved in their definition we  find that, under 
 the residual gauge freedom (\ref{rgfL=1a})-(\ref{rgfL=1b}) (note that  
$(2-2f+r\td_a\td^ar ) r=6M-4Q^2/r$),
 \begin{equation}
 \begin{split}
Z_a^{(\ell=1)}&=C_a^{(\ell=1)}-\frac{r}{2}\td_ah_T^{(\ell=1)}+4E_0A^{(\ell=1)}\td_ar \\
& \to Z_a^{(\ell=1)} + \left(6M-\frac{4Q^2}{r}\right)\td_aX^{(\ell=1)}+4E_0(QX^{(\ell=1)})\td_ar\\
                                       &     = Z_a^{(\ell=1)}+  \td_a\left[\left(6M-\frac{4Q^2}{r}\right)X^{(\ell=1)}\right]
\end{split}
\end{equation}
And then:
\begin{equation} \label{l1gta}
\xi^{(\ell=1)} \to \xi^{'(\ell=1)}=\xi^{(\ell=1)} + \left(6M-\frac{4Q^2}{r}\right)X^{(\ell=1)}.
\end{equation}
Also 
\begin{equation}
\A^{(\ell=1)}  \to {\A^{(\ell=1)}}' =\A^{(\ell=1)} + Q X^{(\ell=1)}. \label{l1gtb}
\end{equation}
{\em A priori}, equation (\ref{l1gtb}) 
 does not imply  that $\A$ is pure gauge, since the $X^{(\ell=1)}$ field is not arbitrary but restricted to 
the condition (\ref{rgfL=1b}). As we will see, the situation is quite subtle. \\

Equations (\ref{l1gta})-(\ref{l1gtb}) suggest that, for $\ell=1$,  
we replace in the LEME (\ref{3.24})-(\ref{3.23}) $\xi^{(\ell=1)}$ and $A^{(\ell=1)}$  by the gauge invariant field 
\begin{equation}\label{fi1}
\varphi := \frac{Q }{r\left(2-2f+r\td_a\td^ar \right)}\xi^{(\ell=1)} - \A^{(\ell=1)} =: \sum_m 
\varphi^{(m)} S_{(\ell=1,m)}.
\end{equation}
If we rewrite (\ref{3.23}) -(\ref{3.24}) in terms of $\varphi$ and $\A$, eliminate second order $\A$ derivatives from 
(\ref{3.23}) using (\ref{3.24}), we get a decoupled equation for $\varphi$:
\begin{equation} \label{LEME1a}
[ -f \td^a \td_a + V^{(\ell=1)} ] \varphi =0,
\end{equation}
where
\begin{equation} \label{V11}
 V^{(\ell=1)}  = - \frac{2 f}{3 r^4 (3Mr -2 Q^2)^2} 
\left( (4\, Q^4 \; \Lambda-27\,M^2) r^4 + 54M^2Q^2 r^2-48M Q^4 r+12 Q^6 \right)
\end{equation}
For $\A^{(\ell=1,m)}$ we obtain 
\begin{equation}\label{LEME1b}
r^{-2} \td^c(r^2 \td_c \A^{(\ell=1,m)}) = 2r^{-1} \td^c r \td_c \varphi^{(m)}   + Z(r) \varphi^{(m)} .
\end{equation}
with
\begin{equation}
Z(r) = \frac{-4Q^2\Lambda r^4+18Mr^3-24Q^2Mr+12Q^4}{3r^4(3Mr-2Q^2)}.
\end{equation}
$\varphi$ is a physical (gauge invariant) degree of freedom obeying (\ref{LEME1a}). 
Once a solution of this equation 
is picked, the source on the right side of (\ref{LEME1b}) is defined, and the solution  of (\ref{LEME1b}) will be  unique 
up to a solution of the homogeneous equation. However, since the homogeneous equation agrees with (\ref{rgfL=1b}), 
in view 
of (\ref{l1gtb}), any two solutions of  (\ref{LEME1b}) are gauge related and therefore equivalent. This implies that
{\em  the 
gauge class} of $\A^{(\ell=1,m)}$ is uniquely determined once the three gauge invariant functions on the orbit space  
$\varphi{(m)}$ 
are given, and then 
$\varphi$ contains 
the only degrees of freedom in the $\ell=1$ subspace (three functions defined on the orbit space).\\

This situation should be contrasted with that of the projections on 
the higher harmonic subspaces $\ell \geq 2$, for which the number of 
degrees of freedom is {\em two} (instead of one) functions on the orbit space for every $(\ell,m)$: the harmonic 
components solutions $\Phi_n^{(\ell,m)}$ of the Zerilli fields $\Phi_n$, $n=1,2$ (see next Section). It should also be 
contrasted with the Schwarzschild black hole case, for which the even $\ell=1$ 
mode is pure gauge \cite{Dotti:2016cqy}.

\subsubsection{Solution of the $\ell>1$ LEME}

To decouple the system (\ref{3.24})-(\ref{3.23})  we introduce 
\begin{align} \label{k1}
\kappa_1 &= -\left(\frac{1}{\j2+2f-r\td_a\td^ar}\right)\sqrt{-(\j2+2)}\;  \xi, \\
\kappa_2 &= -\frac{2Q}{r}\left(\frac{ 1 }{\j2+2f-r\td_a\td^ar}\right)\;\xi^{(\ell \geq 2)}  - 2 \A. \label{k2}
\end{align}
$\j2$ acts as a $-\ell(\ell+1)$  factor on the $\ell$ subspace of $L^2(S^2)$ then, e.g., if 
$\xi = \sum_{(\ell,m)} \xi^{(\ell,m)} S_{(\ell,m)}$ is the expansion of $\xi$ in  spherical harmonics 
$S_{(\ell,m)}$,  the linear operator   in the definition of $\kappa_1$ above acts as 
\begin{equation}
-\left(\frac{1}{\j2+2f-r\td_a\td^ar}\right)\sqrt{-(\j2+2)} \sum_{(\ell,m)} \xi^{(\ell,m)}S_{(\ell,m)} = 
 \sum_{(\ell,m)} \left(\frac{\sqrt{(\ell+2)(\ell-1)}}{\ell(\ell+1)-2f+r\td_a\td^ar}\right) \xi^{(\ell,m)}  S_{(\ell,m)}.
\end{equation}
Using the fact that on scalar fields $\hd^A \hd_A = \j2$, 
the projection of equations  (\ref{3.24})-(\ref{3.23}) on the $\ell>1$ space can then be written  as 
\begin{equation}  \renewcommand*{\arraystretch}{1.5}
\left( \begin{array}{cc}
-\td^a \td_a +\mathbf{U}-3M \mathbf{W} & 2Q \sqrt{-(\j2 +2)}\; \mathbf{W} \\
 2Q \sqrt{(-\j2+2)}\; \mathbf{W}  & -\td^a \td_a +\mathbf{U}+3M \mathbf{W}
\end{array} \right) \; \left(\begin{array}{c} \kappa_1 \\ \kappa_2 \end{array} \right) =0 \label{mop}
\end{equation}
where  $\mathbf{U}$ and $\mathbf{W}$ entering the  symmetric matrix operator $\mathcal{O}$ in  above are 
defined by 
\begin{align}
\left[ r^2\left(\j2+2f-r\td_a\td^ar+r^2\Lambda\right)^2 \right] \mathbf{U} = & -(\j2+2)^3+
\left(2+\frac{9M}{r}-\frac{4Q^2}{r^2}\right)(\j2+2)^2 + \nonumber \\ 
&\left(\frac{3M}{r}+\frac{9M^2+2Q^2}{r^2}-\frac{16Q^2M}{r^3}+\frac{6Q^4}{r^4}+\frac{2\Lambda Q^2}{3}\right)(\j2+2) \nonumber \\
&+4\left(\frac{9M^2}{r^2}+\frac{9M^3}{r^3}-\frac{39Q^2M^2}{r^4}+\frac{32Q^4M}{r^5}-\frac{8Q^6}{r^6}\right) \nonumber \\
&-\frac{4r^2\Lambda}{3}\left(\frac{9M^2}{r^2}-\frac{12Q^2M}{r^3}+\frac{8Q^4}{r^4}\right) 
\end{align}

and 
\begin{equation}
\left[ r^3\left(\j2+2f-r\td_a\td^ar+r^2\Lambda\right)^2 \right]  \mathbf{W} =  (\j2+2)^2 -4(\j2+2)+\frac{4M}{r}\left(3-\frac{3M}{r}+\frac{Q^2}{r^2}\right)
+\frac{4\Lambda}{3}\left(3Mr-4Q^2\right).
\end{equation}
The matrix 
$\mathcal{O}$   can be diagonalized 
 by introducing $\Xi=\sqrt{9M^2-4Q^2(\mathbf{J}^2+2)}$, $\beta_n=3M+(-1)^n\Xi$, $n=1,2$ and
 \begin{equation} \label{P} \renewcommand*{\arraystretch}{1.5}
 P = \left( \begin{array}{cc} -\beta_1 & \beta_2 \\ 2Q \sqrt{-(\j2+2)} & -2Q \sqrt{-(\j2+2)} 
 \end{array} \right).
 \end{equation}
 We find that  
 \begin{equation} \label{pop}
 P^{-1} \mathcal{O} P = \left( \begin{array}{cc} -\td^a \td_a + U_1 & 0 \\ 0 & -\td^a \td_a +U_2 \end{array}
 \right),
 \end{equation}
 where 
 \begin{equation} \label{zeri}
 U_n = \mathbf{U} +\frac{(-1)^{n+1}}{2}(\beta_2-\beta_1) \mathbf{W}.
 \end{equation}
In view of (\ref{pop}),  the Zerilli fields 
 \begin{equation} \label{Pk}
 \left( \begin{array}{c}  \Phi_1 \\ \Phi_2 \end{array} \right) = P^{-1} \; \left( \begin{array}{c}  \kappa_1 \\ \kappa_2 \end{array} 
 \right)
 \end{equation}
 satisfy the equations 
 \begin{equation} \label{Z1}
[ -f \td^a \td_a + V_n ] \Phi_n =0, \;\; n=1,2, 
 \end{equation}
 where $V_n = U_n/f$ can be written in Ricatti form
 \begin{equation} \label{Ricatti}
 V_n=f\beta_n\partial_r f_n+\beta_n^2f_n^2+\mathbf{J}^2(\mathbf{J}^2+2)f_n, 
 \end{equation}
with 
\begin{equation} \label{fn}
f_n=\frac{f}{\left(r\beta_n-r^2\left(\mathbf{J}^2+2\right)\right)}.
\end{equation}
In $t-r$ coordinates (\ref{Z1}) reads
\begin{equation} \label{ZE2}
\p_t^2 \Phi_n+A_n \Phi_n=0
\end{equation}
where
\begin{equation} \label{AE}
A_n=- \p_{r^*}^2 + V_n.
\end{equation}
and $r^*$ is a tortoise coordinate, defined by $dr^*/dr=1/f$.\\
Since $\xi$ and $\A$ are gauge invariant fields, so are $\kappa_1, \kappa_2,$ 
and the Zerilli fields $\Phi_1$ and $\Phi_2$.\\

Tracing our definitions back we find that 
\begin{align} \label{zeri1}
 \A&=-\frac{Q}{r^2}\left[\left(r\beta_1-r^2\left(\mathbf{J}^2+2\right)\right) \Phi_1-
 \left(r\beta_2-r^2\left(\mathbf{J}^2+2\right)\right) \Phi_2\right]\\ \label{zeri2}
 \xi&=\left(\mathbf{J}^2+2f-r\td_a\td^ar\right)\left(\beta_1\Phi_1-\beta_2\Phi_2\right) 
\end{align}
and that the LEME (\ref{3.24}) (\ref{3.23}) reduce to the decoupled  Zerilli equations (\ref{Z1}). 
If we replace $\j2 \to -\ell(\ell+1)$ and 
use $(t,r)$ coordinates on the orbit space, $ [-f \td^a \td_a] \Phi$ reads $\p_t^2 \Phi+ f \p_r(f\p_r \Phi)$ 
and (\ref{Z1}) gives the Zerilli equation for the $\Phi_n^{(\ell,m)}$ harmonic components of $\Phi_n$, $n=1,2$, 
as found for $\Lambda=0$ by Zerilli in \cite{Zerilli:1974ai} and 
Moncrief in \cite{Moncrief:1974ng} and for $\Lambda \neq 0$ in \cite{Kodama:2003kk}.

\section{Non-modal linear stability for even perturbations} \label{nmsSect}

 From the results of the previous Section follows that 
the set $\mathcal{L}_+$ of  equivalent classes $ [(h_{\a \b},F_{\a \b})]$ of even 
 solutions $(h_{\a \b},\f_{\a \b})$ of the LEME 
 mod the Maxwell and the diffeomorphism  gauge equivalence relation (\ref{prime}),   
 can be parametrized by the first order variation $\d M$ and $\d Q$ of 
the mass  $M$ and charge $Q$ ($\ell=0$ modes), the $\ell=1$ field   $\varphi =\sum_m  \varphi_m S^{(\ell=1,m)}$,  
$\varphi_m: \mathcal{N} \to \mathbb{R}, m=1,2,3$ 
satisfying (\ref{LEME1a}), and the Zerilli fields $\Phi_n: \sum_{(\ell \geq 2, m)} \Phi_n^{(\ell,m)}
 S^{(\ell,m)}: \mathcal{M} \to \mathbb{R}$, $n=1,2$ 
(alternatively $\Phi_n^{(\ell ,m)}: \mathcal{N} \to \mathbb{R}$) obeying 
(\ref{Z1}) ($\ell \geq 2$ modes):
\begin{align}\nonumber
\mathcal{L}_+ &= \{ [(h_{\a \b},F_{\a \b})] \; | \; (h_{\a \b},F_{\a \b}) \; \text{ is a solution of the LEME } \} \\
 &= \{ (\d M, \d Q, \varphi, \Phi_n)  \; | \; \varphi \; \text{ satisfies (\ref{LEME1a}) and } 
   \Phi_n, n=1,2 \; \text{ satisfy (\ref{Z1})} \}.  \label{L+}
\end{align}
Although these fields and constants  measure the effects of the perturbation, 
there is a distinction 
 between the $\ell=0$ constants $\d M$ and $\d Q$, which have a clear physical meaning as mass and charge shifts 
  within the Kerr-Newman (A)dS  family, 
and the  $\ell\geq 1$ fields $\varphi, \Phi_1$ and $\Phi_2$. The latter  
 are  convenient to disentangle  the $\ell \geq 1$ LEME but  have, a priori,  no direct 
physical interpretation. I the following Section we will find scalar fields that substitute these and have 
a direct geometrical meaning.

\subsection{Measurable effects of the perturbations}  \label{invs}

There are sixteen real  algebraically independent
basic set of  scalars made out of the Riemann 
tensor in the Carminati-McLenaghan 
  \cite{carmi} basis. Any other scalar field made out of contractions of the tensor product of 
any number of Riemann tensors, volume form and metric tensor can be written as a polynomial 
on these basic scalars. Among these  there are 
six real fields (we follow the notation in \cite{carmi}):
\begin{equation} \label{invs1}
\{ R, r_1 , r_2 , r_3,  m_3, m_4 \}
\end{equation}
and   the five complex fields 
\begin{equation}\label{invs2}
\{w_1, w_2, m_1, m_2, m_5 \}.
\end{equation}
In the electro-vacuum case,   they are constrained by the following (seven real) sizygies \cite{carmi} : 
\begin{equation}
R=0, \; r_2=0, \; 4 r_3- r_1^2, \; m_4=0, \; m_1 \bar m_2 -r_1 \bar m_5=0, \;
m_2 \bar m_2 m_3 - r_1 m_5 \bar m_5=0,
\end{equation}
which leave $r_1,w_1,w_2,m_1,m_2$ as independent fields in the general electro-vacuum case. 
Note that these constraints do not define a manifold but an algebraic variety: the dimension of the tangent space 
(defined by the linearization of the constraints) may change at 
different points. \\

We may also consider invariants involving the Maxwell fields, as well as mixed invariants such as 
($C_{\a \b \g \d}$ the Weyl tensor)
\begin{equation}
F = F_{\a \b}F^{\a \b}, \;\; C= C_{\a \b \g \d} F^{\a \b} F^{\g \d}.
\end{equation}
Due to the symmetries of the background, it can be proved  that the imaginary part of first 
order variations  of the complex scalars $\delta w_1$, ...., $\delta m_5$ vanish trivially under even perturbations, so 
we will focus our attention on the first order variations 
\begin{equation}\label{df}
\d r_1,\;\; \d \Re w_1= \d w_1, \; \;
\d \Re w_2=\d w_2, \;\; \d \Re m_1= \d m_1,  \;\; \d \Re m_2= \d m_2, \d F, \d C.
\end{equation}
Note that \cite{carmi}
\begin{equation}
 w_1=\Re \w_1 =\tfrac{1}{8} C^{\a \b \g \d}C_{\a \b \g \d}. 
\end{equation}
Since the background values of these fields 
\begin{equation}\label{back1}
\begin{split}
{r_1}_o &= \frac{Q^2}{r^8}, \;\; {w_1}_o = \frac{6(Q^2-Mr)^2}{r^8}, \;\; {w_2}_o = \frac{6(Q^2-Mr)^3}{r^{12}}, 
F_o = - \frac{2Q^2}{r^4}\\
{m_1}_o &=\frac{2Q^4(Q^2-Mr)}{r^{12}}, \;\; {m_2}_o =\frac{4Q^4(Q^2-Mr)^2}{r^{16}}, \;\;
C_o = \frac{8Q^2(Q^2-Mr)}{r^8} \\
\end{split}
\end{equation}
do not vanish, none of the fields in (\ref{df}) is gauge invariant. However, it is possible to construct 
gauge invariant fields out of them, such as
\begin{equation} \label{s}
\s:= {w_1}_o' \d C -{C_o}' \d w_1,
\end{equation}
 etc, where the prime denotes derivative with respect 
to $r$.  Under a gauge transformation along $X^{\a}$ 
\begin{equation}
\s \to \s + {w_1}_o'  \; X^{\a} \p_\a C_o -C_o' \; X^{\a} \p_{\a} {w_1}_o = \s,
\end{equation}
since ${w_1}_o$ and $C_o$ (as every curvature scalar) depend only on $r$. This idea generalizes as follows: let 
$I_{(1)},...,I_{(s)}$ be a set of scalar curvature. 
Then $S=\sum_k f_k \d I_{(k)}$ is gauge invariant as long as the $f_k$ satisfy $\sum_k f_k {I_{(k)}}_o'=0$, since  
for a gauge transformation along $X^{\a}$,  
$\d I_{(j)}= X^r {I_{(j)}}_o'$ and $\d S =0$. For $s=2$ this 
reduces to  
\begin{equation} \label{K}
S = K ({I_{(1)}}_0' \d I_{(2)} - {I_{(2)}}_0' \d I_{(1)}) \\
\end{equation}

 When  
calculating the $\ell \geq 2$ projection of 
the first order even perturbation of the fields (\ref{df}) {\em in the RW gauge (\ref{pertrw1})}  
in terms of the Zerilli fields, we get expressions involving up to five derivatives of the $\Phi_n$. 
On shell, that is, assuming the LEME, we can use the Zerilli equation (\ref{Z1}) and its $r-$derivatives repeatidly and, 
 after  lengthy manipulations,  obtain 
simpler expressions 
involving only the $\Phi_n$ and their first $r-$ derivatives. 
We proved that, on shell,  
all the gauge invariant combinations 
 of the first order variation of the fields (\ref{df}) are proportional to each other. 
In other words,  there is a single 
independent gauge invariant combination of first order variation of curvature scalars. This  certainly could not 
carry the same information as the {\em two} fields $\Phi_n, n=1,2$.
Since all the algebraic gauge invariant curvature variation scalars are proportional on shell, it is irrelevant 
to our purposes of a non-modal approach which one we choose. For the field $\s$ in (\ref{s}) we found, after lenghty calculations with the help of 
symbolic manipulation programs, 
\begin{align}\label{s>1} 
\s^{(\ell>1)} &= \frac{96 Q^2 \j2 (\j2+2) (Mr-Q^2)^2}{r^{17}} \left((r \beta_2-4Q^2) \; \Phi_1- (r \beta_1-4Q^2) \; \Phi_2\right), \\ \label{s=1}
\s^{(\ell=1)} &=     \frac{384 \, Q (3Mr-2Q^2)(Mr-Q^2)^2}{r^{17}} \; \varphi, 
\end{align}
and 
\begin{equation}
\begin{split}\label{s=0}
\s^{(\ell=0)} &= \d M \left( {w_1}_o' \; \p_{M} C_o -C_o' \; \p_M {w_1}_o \right)  
+  \d Q \left( {w_1}_o' \; \p_{Q} C_o -C_o' \; \p_Q {w_1}_o \right)  \\
&= \frac{192 \, Q \; (Q^2-Mr)^2}{r^{16}} (3M \; \d Q- 2 Q \; \d M).
\end{split}
\end{equation}
It is an interesting fact that first order $r-$derivatives of the $\Phi_n$, which are present in $\d w_1$ and $\d C$, 
 cancel out in (\ref{s>1}). Note also that equation (\ref{s=1}) gives a geometrical interpretation for the 
gauge invariant $\ell=1$ field $\varphi$.\\

To construct a second curvature related gauge invariant  field that, together with (\ref{s}), 
allows us to recover  the Zerilli fields, 
 we need to 
consider  differential invariants. These will  (at least) one more derivative of the Zerilli fields. 
When simplifying their on shell form we find that first order derivatives 
do not cancel out (at least, in the many examples that we have worked out). \\
The field we chose is cosntructed as follows: 
define
\begin{equation} \label{invs3}
I =\tfrac{1}{720} (\nabla_{\a} C_{\b \gamma \d \tau})  (\nabla^{\a} C^{\b \gamma \d \tau}), \;\; J=(\nabla_{\a} F_{\b \d})
(\nabla^{\a} F^{\b \d}), 
\end{equation}
whose background values are 
\begin{equation}\label{back2}
\begin{split}
I_o  &= \frac{f(r) }{15 r^{10}} (15 M^2 r^2-36 M Q^2r +22 Q^4), \\
 J_o  &= - \frac{12 Q^2 f(r)}{r^6}, 
\end{split}
\end{equation}
then the gauge invariant field
\begin{equation} \label{T}
\T = I_o' \; \d J - J_o' \; \d I
\end{equation}
has an on shell expression with 
\begin{equation} \label{t>1}
\T^{(\ell>1)} = \Upsilon_1 \Phi_1' +  \Upsilon_2 \Phi_2' + \Omega_1 \Phi_1 + \Omega_2 \Phi_2.
\end{equation}
The operators $\Upsilon_n$ and $\Omega_n$ in (\ref{t>1}) do not admit a simple expression. In any case, all 
we need know about them is that,  
for $\Lambda>0$ and $r_h \leq r \leq r_c$ they are bounded, whereas 
for $\Lambda=0$ and $r \leq r_h$ they are bounded and, as $r \to \infty$, behave as
\begin{equation}
\begin{split}
\Upsilon_n& = 12M Q^2 r^{-14} \; \j2 (\j2 +2) \left[ (-1)^n 5M - \sqrt{9M^2-4Q^2(\j2 +2)}\right] + \mathcal{O}(r^{-15}),\\
\Omega_n &= 12M Q^2 r^{-15} \; \j2 (\j2 +2) \left[ (-1)^{n+1} 7M + \sqrt{9M^2-4Q^2(\j2 +2)}\right] + \mathcal{O}(r^{-16}).
\end{split}
\end{equation}
For $\T^{(\ell=1)}$ we find
\begin{equation}
\T^{(\ell=1)} = \frac{16 Q \, f(r)}{15 r^{18}} C(r) \p_r \varphi + \frac{Q \, f(r)}{15 r^{19} (2Q^2-3Mr)} D(r) \varphi,
\end{equation}
with
\begin{align}\nonumber
 C(r) = &45 \, M^2 \Lambda r^6 - 110 \, MQ^2 \Lambda r^5 + (68 Q^4\Lambda-180 M^2) r^4 + 9M(45M^2+46Q^2)r^3 \\
&- 3Q^2( 379 \, M^2 +80 \, Q^2)r^2 + 1014 \, MQ^4 r- 276 \, Q^6,   \label{c} \\ \nonumber
D(r) = &135\,{M}^{3}\Lambda\,{r}^{7}-387\,{M}^{2}{Q}^{2}\Lambda\,{r}^{6}-2\,M
 \left( -193\,{Q}^{4}\Lambda+270\,{M}^{2} \right) {r}^{5}+ \left( -140
\,{Q}^{6}\Lambda+1215\,{M}^{4}+1431\,{M}^{2}{Q}^{2} \right) {r}^{4}\\ 
& -81
\,M{Q}^{2} \left( 47\,{M}^{2}+16\,{Q}^{2} \right) {r}^{3}+3\,{Q}^{4}
 \left( 1477\,{M}^{2}+144\,{Q}^{2} \right) {r}^{2}-2310\,M{Q}^{6}r+444
\,{Q}^{8}.  
\end{align}
For the $\ell=0$ piece of $\T$:
\begin{equation}
\begin{split} \label{t=0}
\T^{(\ell=0)} &= \d M \left( {I}_o' \; \p_{M} J_o -J_o' \; \p_M {I}_o \right)  
+  \d Q \left( {I}_o' \; \p_{Q} J_o -J_o' \; \p_Q {I}_o \right)  \\
&= \frac{16 f(r)\; Q}{5 r^{18}} (Q \T_M \; \d M  - \T_Q \; \d Q).
\end{split}
\end{equation}
Here
\begin{equation}
\begin{split} \label{t=0b}
\T_Q = 15\,{M}^{2}\Lambda\,{r}^{5}-18\,M{Q}^{2}\Lambda\,{r}^{4}-60\,{M}^{2}{r
}^{3}+27\,M \left( 5\,{M}^{2}+2\,{Q}^{2} \right) {r}^{2} +22\,{Q}^{2}
 \left( Q^2-9M^2 \right)   r+42\,M{Q}^{4}
 \end{split}
 \end{equation}
 and
 \begin{equation}  \label{t=0c}
 \T_M =10\,M\Lambda\,{r}^{5}-12\,{Q}^{2}\Lambda\,{r}^{4}-45\,M{r}^{3}+
 \left( 90\,{M}^{2}+54\,{Q}^{2} \right) {r}^{2}-132\,M{Q}^{2}r+28\,{Q}^{4}.
\end{equation}

Equations (\ref{s>1})-(\ref{s=0})  and (\ref{t>1})-(\ref{t=0}) allows us to prove that $\s$ and $\T$ contain 
all the gauge invariant information of a given  
perturbation.

\begin{thm} \label{1} Consider the set of gauge classes of even solutions $[(h_{\a \b},\f_{\a \b})]$ of the 
LEME around a Reissner-Nordstr\"om  
(A)dS black hole background, and the perturbed fields $\s$ and $\T$ defined above. 
The map $[(h_{\a \b}, \mathcal{F}_{\a \b})] \to (\s,\T)$ is injective: it is possible to reconstruct 
a representative 
of $[h_{\a \b}]$ and $[\f_{\a \b}]$  from $(\s, \T )$. 
\end{thm}
\begin{proof}
Assume $\s=0=\T$, then equations (\ref{s=0}) and (\ref{t=0})-(\ref{t=0c}) imply $\d M=0=\d Q$. 
Equation (\ref{s=1})  implies $\varphi=0$, and the combination of (\ref{s>1}) and(\ref{t>1}) gives
$\Phi_1=0=\Phi_2$. This last assertion follows from a reasoning on the line of the proof of Theorem 5 
in \cite{Dotti:2016cqy}: from $\s^{(\ell >1)}=0$ and (\ref{s>1}) we may write $\Phi_2$ in terms of $\Phi_1$ which, 
inserted 
 in the equation $\T^{>1}=0$ using  (\ref{t>1}), gives  an equation for $\Phi_1$ whose only solution compatible 
with (\ref{Z1}) is the trivial one. 
Thus, an electro-gravitational perturbation must be trivial if $\s = 0 = \T$.\\
To reconstruct the perturbation from $\s$ and $\T$ we proceed as in Theorem 1.i in \cite{julian}. 
\end{proof}

 The fact that the $\ell=0$ degrees of freedom are $\d Q$ and $\d M$ explains why 
these quantities can be obtained from $\s_0$ and $\T_0$ by inverting (\ref{s=0}) and (\ref{t=0}). 
In \cite{Ferrando:2002dq}, a characterization of subclasses of type-D spacetimes is made in terms of 
equations involving curvature tensors and scalars. In particular, two curvature scalars are given such that, 
when evaluated on a \rn spacetime, they give the mass and charge (see Theorem 5). Since these scalar fields 
are constant on  \rn backgrounds, their first order perturbation are gauge invariant For $\ell=0$ perturbations they agree exactly 
with $\d M$ and $\d Q$, because these are perturbations along the \rn 
family. The scalar fields in \cite{Ferrando:2002dq} are made out of rational functions of rational powers 
of the basic polynomial invariants (\ref{invs1}),  (\ref{invs2}) and (\ref{invs3}). To illustrate the relation 
between these (a priori) more general scalar gauge invariants and the  ones we constructed above 
we consider the case of scalar perturbations 
of an uncharged (Schwarzschild) black hole. In this case we get from (\ref{back1}) and (\ref{back2})
\begin{equation}
M^2 = \frac{9}{2} \frac{ {w_1}_o^4}{\Lambda {w_1}_o + \sqrt{6} \, {w_1}_o^{3/2} + 3 I_o}.
\end{equation}
Thus, for 
\begin{equation}
Z = \frac{9}{2} \frac{ {w_1}^4}{\Lambda {w_1} + \sqrt{6} \, {w_1}^{3/2} + 3 I},
\end{equation}
$\d Z$ is  gauge invariant, and so is 
\begin{equation}
-r^{-5} \d Z = (9M-4r+\lambda r^3) \frac{\d \w_1}{6} + 3 r^3 \d I = \frac{r^{10}}{12 \, M^2} \left( I_o' \d w_1 - {w_1}_o' \d I   \right), 
\end{equation}
which is of the form (\ref{K}) and 
agrees with the gauge invariant field $G_+$ used in the analysis of the even Schwarzschild perturbations in 
\cite{Dotti:2016cqy} (equation (202)).

\subsection{Pointwise boundedness of $\Phi_n$ and  $\p_{r^*}\Phi_n$} \label{pb}

In this section we restrict our attention to black hole solutions with horizons $0<r_i<r_h<r_c$. The relation between the
radii of the horizons and $\Lambda, M$ and $Q$ can be easily obtained from (\ref{f}) and (\ref{ff}). 
For $\Lambda>0$ [$\Lambda=0$] we are interested in the range $r_h<r<r_c$ [$r>r_h$]. In 
both cases the tortoise radial coordinate satisfies  $-\infty < r^* < \infty$. \\
\begin{figure}[htb]
    \includegraphics[width=.3\textwidth]{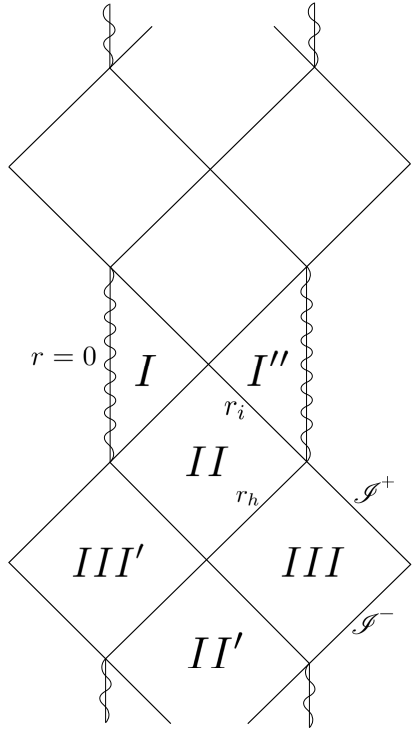}
  \caption{\label{rnfig} The Carter-Penrose diagram of (part of) the maximal analytic extension of the $|Q|<M$ 
\rn  black hole. 
The union of II, II', III and III'   is globally hyperbolic, its boundary at $r_i$ is a Cauchy horizon.}
\end{figure}

\begin{figure}[htb]
    \includegraphics[width=.6\textwidth]{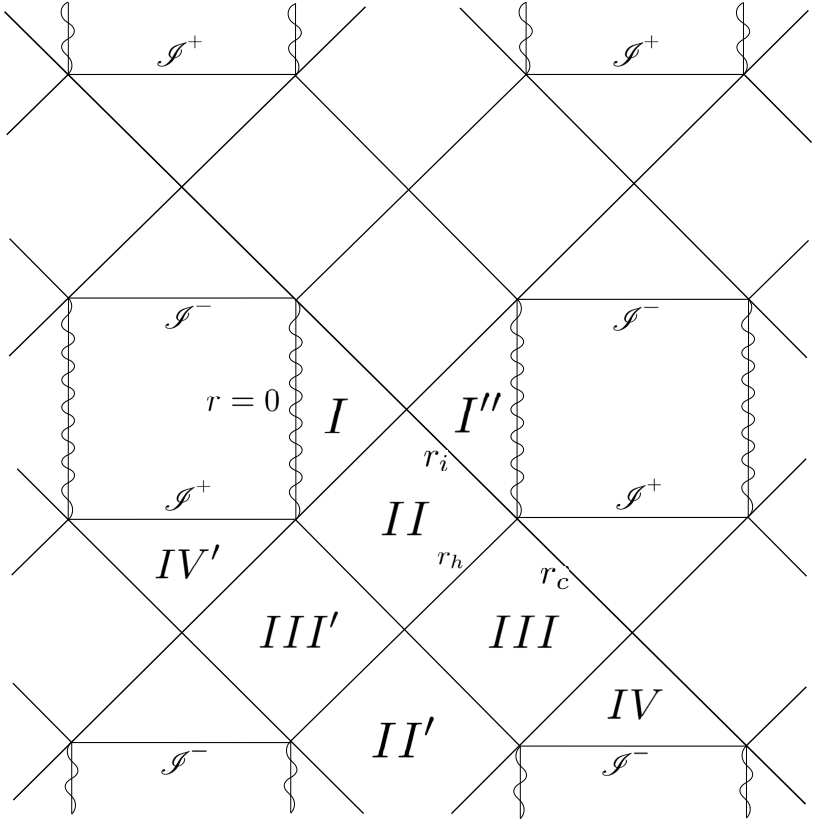}
  \caption{\label{rndsfig} The Carter-Penrose diagram of (part of) the maximal analytic extension of a non extremal 
(three different horizons) 
\rn  de Sitter black hole. }
\end{figure}

\begin{thm}\label{bound}  \label{2} 
Assume $\Phi_n$ is a smooth solution of equation (\ref{ZE2}) on the union of regions II, II', III and III' 
of the  extended \rn (figure \ref{rnfig}) or \rn   de Sitter (figure \ref{rndsfig}) spacetimes, 
 with compact support on Cauchy surfaces. There exist constants $C_o, L_o$ that depend on the 
datum of this field at a Cauchy surface, such that $|\Phi_n|<C_o$ and $|\p_{r^*} \Phi_n| < L_o$ 
for all points in the outer static region $III$. 
\end{thm}
\begin{proof} 
As in the proof of Theorem 2 in \cite{julian},
 and following \cite{Kay:1987ax}, we may restrict, without loss of generality, to fields 
that vanish on the bifurcation sphere together with their Kruskal time derivatives (for details, see \cite{Kay:1987ax}).\\
On a  $t$ slice of region III, define the $L^2$ norm of a real field $G$ as 
\begin{equation} \label{norm}
\lVert G \rVert^2=\langle G|G\rangle=\int_{\mathbb{R} \times S^2}  G^2 \;  dr^* \sin(\theta) d\theta d\phi, \;\; dr^* = \frac{dr}{f}.
\end{equation}
Using a Sobolev type inequality (equation (5.27)  in  \cite{dk}) on the Zerilli fields $\Phi_n$ at a fixed time $t$ gives:
\begin{equation} \label{sobo}
|\Phi_n(t,r^*,\theta,\phi)| \leq  C \left(\left. \lVert \Phi_n \right|_t \rVert + \left. \lVert \p_{r^*}^2 \Phi_n \right|_t \rVert +
\left. \lVert \j2 \Phi_n \right|_t \rVert \right). 
\end{equation}
where $C$ is a constant.  
We will follow the  strategy in \cite{Kay:1987ax} of proving that the $L^2$ norms 
on the right hand side of (\ref{sobo}) can be bounded by the energies of related field configurations. Since 
energy is conserved for solutions of (\ref{ZE2}), we get in this way a $t-$independent upper bound of the 
right side of (\ref{sobo}) and therefore,  a global bound of $|\Phi_n(t,r^*,\theta,\phi)|$ for all $(t,r^*,\theta,\phi)$, 
i.e., of $\Phi_n$ in the outer static region III. \\

The inner product defined by the norm (\ref{norm}), simplifies, after 
introducing an expansion in real orthonormal spherical harmonics (e.g., tesseral spherical harmonics) $S_{(\ell,m)}$. If  
$G = \sum_{(\ell,m)} g_{(\ell,m)} S_{(\ell,m)}$ and  $K = \sum_{(\ell,m)} k_{(\ell,m)} S_{(\ell,m)}$ then 
\begin{equation}
\langle G|K\rangle= \sum_{(\ell,m)} \int_{\mathbb{R}}   g_{(\ell,m)} k_{(\ell,m)} \;  dr^*. 
\end{equation}

From (\ref{AE}) we  get:
\begin{equation} \label{AE2}
\lVert \p_{r^*}^2 \fn \rVert \leq \lVert A_n \fn \rVert + \lVert V_n\fn \rVert
\end{equation}
where $V_n$, given in (\ref{Z1}), can be written as
\begin{align} \label{decomp}
V_n=\frac{{}_n{Z_1}}{D_n}\,\beta_n+\frac{{}_n{Z_2}}{D_n^2}\,\beta_n^2
+\frac{{}_n{Z_3}}{D_n}\,\j2(\j2+2)
\end{align}
being
\begin{equation}
 {}_n{Z_1}= f (f/r^2)' = -\frac{2f}{r^4} \left(r-3M+\frac{2Q^2}{r}\right),\,\,\,
 {}_n{Z_2}=\frac{2f^2}{r^4},\,\,\,\,
 {}_n{Z_3}=\frac{f}{r^2},\;\;
D_n =\beta_n/r-\left(\mathbf{J}^2+2\right)
\end{equation}
Note that: i)  The ${}_n{Z_j}, j=1,2,3$,  depend only on $r$ and are bounded in the domain of interest 
$r_h \leq r \leq r_c$ if $\Lambda>0$ [$r>r_h$ if $\Lambda=0$] by constants  $ {}_n{z_j} > |{}_n{Z_j}(r)| $ 
that depend on $M,Q$ and $\Lambda$; ii) $\Phi_n^{(1)} := \beta_n \Phi_n$, $\Phi_n^{(2)} := \beta_n^2 \Phi_n$ 
and $\Phi_n^{(3)} := \j2(\j2+2) \Phi_n$ are solutions of the Zerilli equation (\ref{Z1}) if $\Phi_n$ is 
a solution; also $\Phi_n^{(4)} :=A_n \Phi_n = -\p_t^2 \Phi_n$ is a solution. This is so because 
any operator that is a function of $\j2$ and $\p_t$, commutes with the operator (\ref{Z1}); iii) 
On the $\ell$ eigenspace of $\j2$, $\ell=2,3,4,...$, $D_n$ acts multiplicatively as
\begin{equation}
D_n^{\ell}(r) =\beta_n/r+(\ell-1)(\ell+2) 
\end{equation}
For $r_h \leq r \leq r_c$ [$r>r_h$ if $\Lambda=0$] and $\ell \geq 2$, $|D_2^{\ell}(r)|$ has 
an absolute minimum $D_2^*>0$ at $r=r_h$ and $\ell=2$, whereas $|D_1^{\ell}(r)|$ also has 
a nonzero absolute minimum $D_1^*$ (possibly at an $\ell>2$). 
Then the first term of $V_n \Phi_n$ 
(see (\ref{decomp})) can be bounded as follows
\begin{align}
\left \lVert \frac{{}_n{Z_1}}{D_n} \b_n  \left. {\Phi}_n \right|_t \right \rVert^2&=\int_{\mathbb{R} \times S^2} \left[ \sum_{\ell m} 
 \frac{{}_n{Z_1}}{D_n}  \left. ({\Phi}_n^{(1)})^{ (\ell,m)} \right|_t S_{(\ell,m)} \right]^2 \;  dr^* \sin(\theta) d\theta d\phi \\
&\leq \left(\frac{{}_n{z_1}}{D_n^*} \right)^2\int_{\mathbb{R}}   \sum_{\ell m}  \left[
  \left. ({\Phi}_n^{(1)})^{ (\ell,m)} \right|_t  \right]^2 \;  dr^* \\
&= \left(\frac{{}_n{z_1}}{D_n^*} \right)^2 \lVert \left. {\Phi}_n^{(1)} \right|_t \rVert^2
\end{align}
Proceeding similarly with the other terms in (\ref{AE2})-(\ref{decomp})  and using the triangle inequality 
gives
\begin{equation} \label{AE3}
\lVert \p_{r^*}^2 \fn \rVert \leq  \left(\frac{{}_n{z_1}}{D_n^*} \right) \lVert  {\Phi}_n^{(1)} |_t \rVert
+  \left(\frac{{}_n{z_2}}{D_n^*} \right) \lVert  {\Phi}_n^{(2)} |_t \rVert  +
 \left(\frac{{}_n{z_3}}{D_n^*} \right) \lVert  {\Phi}_n^{(3)} |_t \rVert + 
 \lVert  {\Phi}_n^{(4)} |_t \rVert
\end{equation}
Inserting this in (\ref{sobo2}) gives 
\begin{equation} \label{sobo3}
|\Phi_n(t,r^*,\theta,\phi)| \leq  K' \left(\left. \lVert \Phi_n \right|_t \rVert + 
 \lVert  {\Phi}_n^{(1)} |_t \rVert +  \lVert  {\Phi}_n^{(2)} |_t \rVert+ \lVert  {\Phi}_n^{(3)} |_t \rVert
+ \lVert  {\Phi}_n^{(4)} |_t \rVert+ \lVert  {\Phi}_n^{(5)} |_t \rVert \right)
\end{equation}
where $K'$ is the maximum over $j$ and $n$ of the constants  ${}_n{z_j} K/D_n^*$, and 
$\Phi_n^{(5)} = \j2 \Phi_n$. \\

 The conserved (i.e., $t-$independent) energy associated to equation (\ref{ZE2}) is 
\begin{equation} \label{energy}
E = \frac{1}{2} \int _{\mathbb{R} \times S^2} ((\p_t \Phi_n)^2 + \Phi_n A_n \Phi_n )\;  dr^* \sin(\theta) d\theta d\phi.
\end{equation}
Since $E$ does not depend on $t$, we may regard it as a functional on the initial datum: $E = E(\Phi^o_n, \dot \Phi^o_n)$, where $\Phi^o_n
 = \left.  \Phi_n \right|_{t_o}$ and $ \dot \Phi^o_n=\left. (\p_t \Phi_n) \right|_{t_o}$:
\begin{equation} \label{energyid}
E(\Phi^o_n, \dot \Phi^o_n)= \frac{1}{2} \int _{\mathbb{R} \times S^2} ((\dot \Phi^o_n)^2+ \Phi^o_n A_n \Phi^o_n
 )\;  dr^* \sin(\theta) d\theta d\phi.
\end{equation}
Using the facts that: i) $A_n, n=1,2,$ is positive definite in the cases we are interested in 
(proved by means of an S-deformation in \cite{Kodama:2003kk}) and so $A_n^{\pm 1/2}$ can be defined 
by means of the spectral theorem (as well as any other power of $A_n$); ii) for a solution $\Phi_n$ of  (\ref{ZE2}), 
$A_n^p \Phi_n$ is a solution of  (\ref{ZE2}) and iii) equation (\ref{energyid}), follows that for 
a solution of  (\ref{ZE2})
\begin{equation}\label{cota1}
 \lVert \left.  \Phi_n \right|_t \rVert^2 \leq 2 \, E(A_n^{-1/2} \Phi_n^o, A_n^{-1/2} \dot \Phi_n^o).
 \end{equation}
We now use the fact that applying to a Cauchy datum  $(\Phi_n^o, \dot \Phi_n^o)$ an operator that is a function of $\j2$ or 
$A_n$ commutes with time evolution.  This allows us   to estimate 
each term on the right hand side of (\ref{sobo3})  with the energy of field configurations related to the one with initial datum 
$(\Phi_n^o, \dot \Phi_n^o)$. Let $B^{(j)}\Phi_n:=\Phi_n^{(j)}, j=1,...,5$ (that is, 
for $j=1,...,5$ these operator are respectively $\beta_n\,,\beta_n^2\,, \j2(\j2+2), A_n$ and $\j2$). From (\ref{cota1})
\begin{equation} \label{energy}
 \lVert  \Phi_n^{(j)}|_t  \rVert^2 \leq  2 \; E\left(A_n^{-\frac{1}{2}} B^{(j)} \Phi_n^o,A_n^{-\frac{1}{2}} B^{(j)} \dot \Phi_n^o\right),\\
\end{equation}
Thus, we can replace the right hand side of (\ref{sobo3}) by  time independent constant $C_o$ made out of  
the initial data $(\Phi_n^o, \dot \Phi_n^o)$ (from which the energies of the related fields $B^{(j)} \Phi_n$ can be computed)
\begin{equation}
|\Phi_n| < C_o  
\end{equation}
It is interesting to note why the fields $B^{(j)} \Phi_n$ have finite energy (a fact tacitly used above): we are assuming smooth 
solutions of the LEME, therefore the 
$\Phi_n$ are $C^{\infty}$ on the sphere, and the series $\sum_{\ell m} [\ell (\ell+1)]^k \Phi_n^{(\ell,m)} 
= (-\j2)^k \Phi_n$
converge for any $k$. In particular, the $\Phi_n^{(\ell,m)}$ decay faster than any power of $\ell$. \\

We will also need a $t-$independent bound for $|\p_{r^*} \Phi_n|$. This can be obtained following 
the same ideas in \cite{Dotti:2016cqy}, taken from \cite{Dain:2012qw}. Starting from the Sobolev inequality (cf equation 
(\ref{sobo2})) applied to $|\p_r^* \Phi_n|$
\begin{equation} \label{sobo2}
|\p_{r^*} \Phi_n(t,r^*,\theta,\phi)| \leq  L \left(\left. \lVert \p_{r^*} \Phi_n \right|_t \rVert +
 \left. \lVert \p_{r^*}^3 \Phi_n \right|_t \rVert +
\left. \lVert \j2 \p_{r^*} \Phi_n \right|_t \rVert \right). 
\end{equation}
Now, using the fact that, for $\Lambda \geq 0$ 
and four dimensions,  the $V_n$ are nonnegative in the interval of interest \cite{Kodama:2003kk}
\begin{equation}
 \lVert \p_{r^*}  \Phi_n |_t \rVert^2 \leq \langle \Phi_n | A_n  \Phi_n\rangle \leq 2 E(\Phi_n^o,\dot \Phi_n^o)
 \end{equation}
This places a $t-$ independent bound on the first term on the right of equation (\ref{sobo2}). The third term can 
be similarly bounded with $E(\j2 \Phi_n^o,\j2 \dot \Phi_n^o)$. For the second term we use 
\begin{equation}
\p_{r^*}^3 \Phi_n = - \p_{r^*} (A_n \Phi_n) + \p_{r^*}V_n \Phi_n + V_n \p_{r^*}\Phi_n,
\end{equation}
$\lVert  \p_{r^*} (A_n \Phi_n) \rVert^2 \leq 2 E(A_n \Phi_n^o,\dot A_n \Phi_n^o)$
 and the boundedness of the operator $ \p_{r^*}V_n$. \\
Proceeding as above, we arrive at the desired pointwise bound:
\begin{equation}
|\p_{r^*} \Phi_n| < L_o
\end{equation}
\end{proof}

\subsection{Pointwise boundedness of $\s$ and $\T$} \label{pb2}

We can now proceed to complete the proof of  nonmodal stability  by showing that the scalar fields 
$\s$ and $\T$ are pointwise bounded in the region of interest by constants that depend on 
the initial conditions.

\begin{thm}\label{3}
Under the assumptions of the Theorem \ref{2}, 
in the outer static region III of a $\Lambda \geq 0$ \rn black hole there holds 
\begin{equation} \label{invbounds}
\s < \frac{A_o}{r^{14}}, \;\; \T < \frac{B_o}{r^{14}},
\end{equation}
where $A_o$ and $B_o$ are constants that depend on the Cauchy datum $(\Phi^o_n, \dot \Phi^o_n)$ of the perturbation.
\end{thm}
\begin{proof}
Let us consider the first inequality. Using the facts that $\j2 (\j2+2) \Phi_n$ and 
$\j2 (\j2+2) \beta _k \Phi_n$ are solutions of the Zerilli equation (\ref{Z1}) with an 
energy that is a function of  $(\Phi^o_n, \dot \Phi^o_n)$, Theorem \ref{2} and equation 
(\ref{s>1}), we find  for $\Lambda=0$ that  $|\s^{(\ell >1)}| < \frac{A_o^{\ell >1}}{r^{14}}$ with 
$A_o^{\ell >1}$ a constant that depends on the $\ell >1$ piece of the initial datum
(the inequality holds trivially for $\Lambda >0$ and $r_h < r < r_c$). For the $\ell=1$ piece 
we use equation (19) in the Erratum in \cite{wald}, applied to the harmonic components of $\varphi$. This gives 
$|\varphi| $ less than a constant that depends on the $\ell=1$ piece of the datum. Then, from 
(\ref{s=1}), follows  $|\s^{(\ell =1)}| < \frac{A_o^{(\ell=1)}}{r^{14}}$ with 
$A_o^{(\ell=1)}$ a constant that depends on the $\ell =1$ piece of the initial datum
(once again, the equality holds trivially for $\Lambda >0$ and $r_h < r < r_c$). Finally, 
from equation (\ref{s=0}) follows trivially that $|\s^{(\ell =0)}| < \frac{A_o^{(\ell=0)}}{r^{14}}$
where $A_o^{(\ell=0)}$ is a constant made related to the $\ell=0$ initial data $(\d M, \d Q)$. \\
To prove the second inequality in (\ref{invbounds}) we proceed exactly as above. We  only need 
a proof of the pointwise boundedness for $f(r)\p_r \varphi$, for which we proceed as in 
\cite{Dain:2012qw} (see the paragraph starting at equation (83)) .
\end{proof}

\section{Discussion}

We have shown in Theorem \ref{1} that the gauge invariant curvature related perturbation fields $\s$ and $\T$,  
defined in equations (\ref{s}) and (\ref{T}), contain all the gauge invariant information of an even  perturbation 
class $[(h_{\a \b}, \mathcal{F}_{\mu \nu})]$ around a Reisner-Nordstr\"om (dS) black hole. From these fields, 
a representative $(h_{\a \b}, \mathcal{F}_{\mu \nu})$ of the perturbation in, say, the Regge-Wheeler gauge, can be 
reconstructed (Theorem 1). For smooth perturbations with compact support on a Cauchy surface of (a copy of) 
the union of regions II, II', IIII and III' 
(see figures \ref{rnfig} and \ref{rndsfig}), these fields are pointwise bounded on the outer region 
(equation  (\ref{invbounds}) in 
Theorem \ref{3}). These results, together with those in \cite{julian}, complete the proof of nonmodal 
linear stability of the outer region of a (dS) Reissner-Nordstr\"om black hole. \\

The large $|t|$ decay of the Zerilli fields (see \cite{Price:1971fb}, \cite{Brady:1996za} 
and  the recent decay results by E. Georgi  in \cite{Giorgi:2019kjt} 
 and references therein) 
and  the similarly expected behavior of $\varphi$, 
together with equations (\ref{s>1})-(\ref{s=0}) and (\ref{t>1})-(\ref{t=0}) give
\begin{equation} \label{Flt}
\mathcal{S} \simeq    \frac{192 \, Q \; (Q^2-Mr)^2}{r^{16}} (3M \; \d Q- 2 Q \; \d M) ,
\end{equation}
and 
\begin{equation} \label{Qlt}
\mathcal{T} \simeq     \frac{16 f(r)\; Q}{5 r^{18}} (Q \T_M \; \d M  - \T_Q \; \d Q) 
\end{equation}
as $t \to \infty$, within a bounded range of $r$ (that growths towards the future) in region III 
(the quantities $\T_M$ and $\T_Q$   were defined in  equations (\ref{t=0})-(\ref{t=0c})). 
The inequalities 
 (\ref{invbounds}), instead,  hold on the entire region III.\\
Together with equations (122) and (123) in \cite{julian}, 
equations (\ref{Flt}) and (\ref{Qlt}) indicate that,  
 for large $t$, the perturbed black hole settles into a Kerr-Newman black hole with parameters 
$M + \d M$, $Q+ \d Q$ and $\vec{J} + \d \vec{J}$. \\

The importance of the result in Theorem \ref{2} lies in the possibility 
of analyzing stability and instability effects in terms of the fields $\s$, $\T$ (and $\mathcal{Q}$ and $\mathcal{F}$ 
in \cite{julian}). The divergence of $\frac{d}{d \tau}\mathcal{S}$ and $\frac{d}{d \tau} \mathcal{T}$ for observers crossing the Cauchy 
horizon $r_i$ can be proved in the same way the divergence of $\frac{d}{d \tau}\mathcal{Q}$ and $\frac{d}{d \tau} 
\mathcal{F}$ was proved for the odd sector scalars in Section IV in \cite{julian}. 
Using these four fields, statements such as the Cauchy horizon instability or the event horizon 
transverse derivative instabilities 
 \cite{penrose}-\cite{Lucietti:2012xr}
 acquire a clear geometrical meaning. \\

\section{Acknowledgements}
This work was partially funded by grants PIP 11220080102479
(Conicet-Argentina) and 30720110101569CB (Universidad Nacional de C\'ordoba). 
J.M.F.T. is supported by a fellowship from Conicet. Many calculations were performed making  
use of the grtensor package 
\cite{grtensor}. We thank an anonymous referee for pointing out the need of a topological term in  equation (\ref{dF}), bringing our 
attention to reference \cite{Ferrando:2002dq} and spotting a number of misprints in earlier versions of the manuscript.

\end{document}